\documentclass[10pt, a4paper, onecolumn, conference, compsocconf]{IEEEtran}
\ifCLASSINFOpdf
\else
\fi

\usepackage{times}
\usepackage{graphicx}
\usepackage[cmex10]{amsmath}
\usepackage{amsthm}
\usepackage{amssymb}

\newenvironment{list1}{\begin{list}{$\bullet$}
{\topsep 0 pt \parsep 0 pt \partopsep 0 pt \itemsep 0
pt}}{\end{list}}

\newtheorem{defi}{Definition}
\newtheorem{theo}{Theorem}

\begin{document}
%
\title{FTOS-Verify: Analysis and Verification of Non-Functional Properties for Fault-Tolerant Systems}


\author{
\IEEEauthorblockN{Chih-Hong Cheng\authorrefmark{1},
                    Christian Buckl\authorrefmark{2},
                    Javier Esparza\authorrefmark{3},
                    Alois Knoll\authorrefmark{1}}
\IEEEauthorblockA{\authorrefmark{1}Unit 6: Robotics and Embedded
Systems, Department of Informatics, TU Munich, Germany}
\IEEEauthorblockA{\authorrefmark{2}Fortiss GmbH, Germany}
\IEEEauthorblockA{\authorrefmark{3}Unit 7: Theoretical Computer
Science, Department of Informatics, TU Munich, Germany
\\Email:\{chengch,buckl,esparza,knoll\}@in.tum.de }}


%


\maketitle

\begin{abstract}
The focus of the tool FTOS is to alleviate designers' burden by
offering code generation for non-functional aspects including
fault-tolerance mechanisms. One crucial aspect in this context is
to ensure that user-selected mechanisms for the system model are
sufficient to resist faults as specified in the underlying fault
hypothesis. In this paper, formal approaches in verification are
proposed to assist the claim. We first raise the precision of FTOS
into pure mathematical constructs, and formulate the deterministic
assumption, which is necessary as an extension of Giotto-like
systems (e.g., FTOS) to equip with fault-tolerance abilities. We
show that local properties of a system with the deterministic
assumption will be preserved in a modified synchronous system used
as the verification model. This enables the use of techniques
known from hardware verification. As for implementation, we
develop a prototype tool called FTOS-Verify, deploy it as an
Eclipse add-on for FTOS, and conduct several case studies.
\end{abstract}



%
\IEEEpeerreviewmaketitle

\section{Introduction}

Fault-tolerant systems refer to systems with the ability to
withstand transient or permanent faults; these faults may be
caused by design errors, hardware failures or environmental
impacts. Applications domains are amongst others the medical,
avionic or automation domain. The introduction of fault-tolerance
abilities into embedded design brings two potential issues
compared to standard systems.

\begin{list1}

\item Fault-tolerance mechanisms require redundancy; additional
means (hardware, information, etc.) are not necessary to implement
the actual functionality during fault-free operation. From a
supplier's view, it is always desirable to equip the system with
"just enough" fault-tolerance abilities, i.e., existing mechanisms
should be sufficient for the resistance of faults based on the
underlying fault hypothesis, but extra mechanisms should not be
introduced due to cost reasons. \item In addition, verification is
of special interest in comparison to standard systems, since
failures might lead to severe
      damages or even endangerment of life.
\item Furthermore, extra time required for validation might
postpone the time-to-market schedule.
    \end{list1}
To alleviate the designers' burden on above issues,
in this paper we concentrate on systematic methodologies to
integrate automatic formal verification (in particular, verifying
non-functional properties regarding the validity of
fault-tolerance mechanisms) into the design process of
fault-tolerant systems, since it is regarded as a rigorous method
to guarantee correctness.
We use FTOS 
\cite{buckl:2007:gft} as our target, which is a model-based tool
for the development of fault-tolerant real-time systems. In the
setting of FTOS, the designers can select predefined (or
self-defined) fault-tolerance mechanisms during their design of
their system models, and the corresponding code is generated
automatically by the tool.
Our goal is to step further  by reporting designers a proof
whether the equipped mechanisms are sufficient to resist the fault
as specified in the fault hypothesis.

Our first job is to raise the precision of FTOS into pure
mathematical constructs. To achieve this purpose, we propose a
formal mathematical model called \emph{global-cycle-accurate (GCA)
system} (section~\ref{sec:GCA.Syntax.Semantics}), which captures
the essence of such systems. Intuitively, a GCA system can be
viewed as an extension from Giotto systems
\cite{Henzinger01giotto:a} based on the model of computation
\emph{Logical Execution Time} (see section~\ref{sec:LET} for
concept description) equipped with redundancies. Nevertheless,
challenges for verifying such systems remain.
\begin{list1}
\item First, it is difficult to construct the verification model
due to the inherent behavior of the MoC, because a GCA system is
synchronous in logical (global-tick) level while asynchronous in
action (micro-tick) level.

\item Secondly, intuitive extensions of redundancy break internal
determinism originally maintained in Giotto-like systems.

\item Lastly, mapping from models to different platforms
(synchronous, asynchronous) while preserving the property is
non-intuitive.
\end{list1}

To overcome these challenges, we propose the concept of the
deterministic assumption
(section~\ref{sec:Deterministic.Assumption}), which relates all
deployed platforms with common features. We discuss impacts of the
deterministic assumption in GCA systems
(section~\ref{sec:Impacts.Deterministic.Assumption}). Most
importantly, with deterministic assumption, some properties are
preserving among all deployed platforms, and we can construct a
simplified model for verification, where property checking can be
achieved by examining a smaller set of reachable state space. The
result also holds with additional assumptions regarding the
introduction of faults (section~\ref{sec:Effect.Fault}).

With the knowledge, we implement an Eclipse add-on called
FTOS-Verify, which enables the use of formal verification under
FTOS, and outline our case study
(section~\ref{sec:Implementation},~\ref{sec:Case.Study}. The
template-based approach enables the software and the templates to
be extended easily, for the introduction and automatic
verification for fault-tolerance mechanisms. Also features in
FTOS-Verify enable non-experts in verification to use the tool
with ease. We mention related work and give the conclusion in
section~\ref{sec:Related.Work}
and~\ref{sec:Conclusion.Future.Work}.

\subsection{The Concept of Logical Execution Time\label{sec:LET}}

The concept of fixed logical execution time (FLET), which is
widely accepted by some design tools, e.g., Giotto, TDL
\cite{simmons:1998:tdl}, HTL \cite{ghosal:2006:htl}, is the
primary technique used in FTOS and thus required to be mentioned
forehand. The basic motivation for FLET is based on the synchrony
assumption of synchronous languages, which assumes an infinitely
fast underlying hardware. For a valid implementation, the actions
of a logical moment must be executed before the next logical
moment appears; right before a logical moment, the system behaves
deterministic. However, when observing the behavior at the level
of
individual actions, no assumptions can be made. 

When applying this concept, the designer only specifies the start
time $t_{start}$ and finish time $t_{end}$ (mostly it is
periodic). Contrary to the traditional view that an output should
be ready as soon as possible, the FLET compiler will ensure that
the output is observable at $t_{end}$ but no earlier. Within the
guarantee, when multiple systems execute in one machine, the
compiler can also allow preemptive scheduling. Under the concept
of FLET, \emph{internal determinism} is guaranteed, meaning that
the relative ordering of operations is the factor to derive
deterministic results, irrelevant of real time.
Fig.~\ref{fig:Relative.Ordering.Giotto} illustrates that in
Giotto, two deployments with same relative ordering behave the
same.

\begin{figure}
 \centering
 \includegraphics[width=0.5\columnwidth]{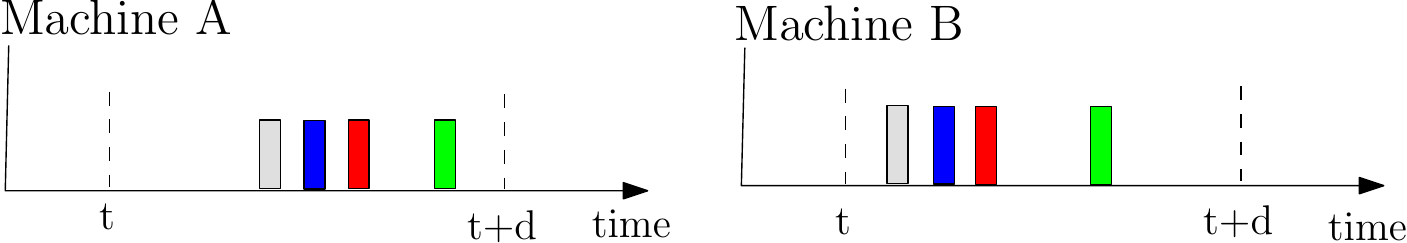}
 \caption{Two Giotto deployments with different scheduling policies.}
 \label{fig:Relative.Ordering.Giotto}
\end{figure}

\section{GCA systems with Redundancies: Syntaxes and Semantics\label{sec:GCA.Syntax.Semantics}}

To describe the FTOS execution model formally, in the following we
define several terms, and introduce a simple mathematical
construct called GCA systems.

\begin{defi}
Define the pattern of the machine with coefficient $(x,n)$ be
$\mathcal{A}_{x,n} = (V \cup V_{env_x}, \overline{\sigma}, Q_x)$.
\begin{list1}
    \item $V$ is the set of arrays with length $n$, and $V_{env_x}$ is the set of environment variables.
    \item  $\overline{\sigma} := \sigma_1 ; \sigma_2 ; \ldots ; \sigma_k$ is a fixed sequence of actions, where for all $j = 1 \ldots k$,
      $ \sigma_j := \verb"send"(a[x])\; | \; a[x] \leftarrow \verb"e" \; |\; \verb"receive"(a[x \oplus 1],\ldots, a[x \oplus n])$ is the atomic action; $\oplus$ is the modulo-plus operator over $n$, $a[x]$ is the x-th element of the array $a$, and $\verb"e"$ is an operation over variables in $V_{env_x} \cup \bigvee_{i \in 1 \ldots n,\; a \in V} a[x \oplus i]$.
    \item $Q_x$ is the message queue.
\end{list1}
\end{defi}

\begin{defi}
Define the redundant system $\mathcal{S}_R$ with n-redundancy be
$\bigvee_{i=1 \ldots n} \mathcal{A}_{i,n}$.
\end{defi}

We explain the intuitive meaning of the formal syntax with the
assistance of fig.~\ref{fig:Parameterized}. In the definition of a
pattern, $n$ is used to represent the number of redundancies
deployed, and $x$ is the index of the machine. Fault tolerance is
mostly achieved by voting of the same variable from different
machines, and in FTOS it is done distributively in each machine.
Therefore, for the array $b[1,\ldots,n]$, in machine $x$ we use
$b[x]$ to store its own value, while other variables in array
$b[1,\ldots,n]$ are used as the stored copy sent by other
machines. In fig.~\ref{fig:Parameterized}, sequence of actions are
instantiated based on different machine index from the actions
specified in the pattern.

\begin{figure}
 \centering
 \includegraphics[width=0.6\columnwidth]{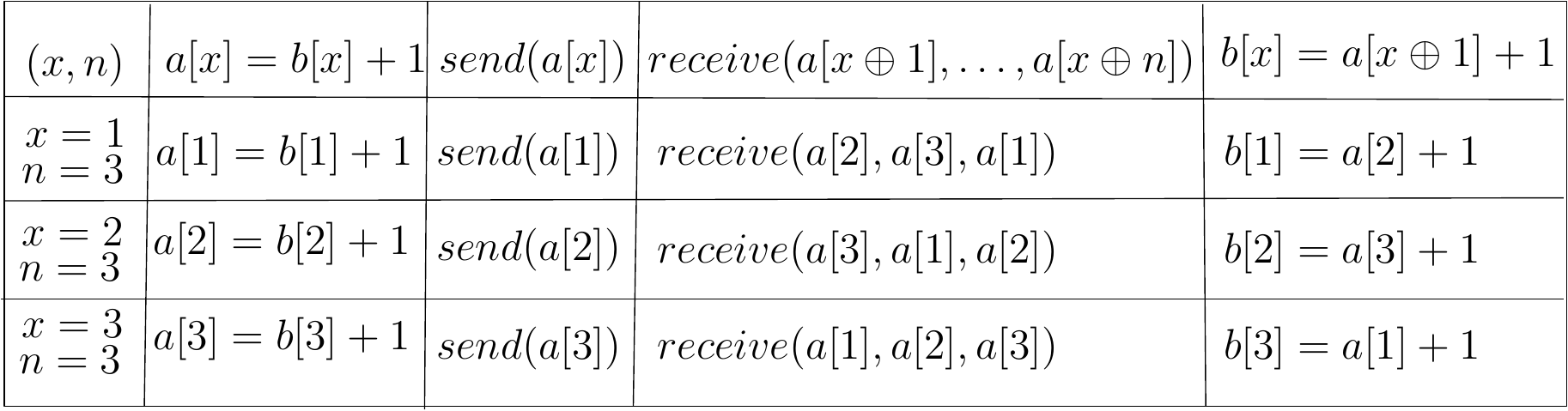}
 \caption{Parameterized actions defined in the pattern, and an instantiation with $n=3$.}
 \label{fig:Parameterized}
\end{figure}

\begin{defi}
The configuration of $\mathcal{S}_R=\bigvee_{i=1 \ldots n}
\mathcal{A}_{i,n}$ is $((v_1, q_1, \Delta_{next_1}), \ldots,
(v_n,q_n,\Delta_{next_n}))$. For each machine $i= 1\ldots n$,
\begin{list1}
    \item $v_i$ is the set of the current values for the variable set $V$.
    \item $q_i$ is the current content for the local message queue $Q_i$.
    \item Let $\verb"atomic"(\overline{\sigma})$ be the set of atomic operations in the action sequence $\overline{\sigma}$, then $\Delta_{next_i} \in \verb"atomic"(\overline{\sigma}) \cup \{\verb"null"\}$ is the next atomic action taken in $\overline{\sigma}$.
\end{list1}
\end{defi}

\begin{defi} The change of configuration $\mathcal{S}_R=\bigvee_{i=1 \ldots n} \mathcal{A}_{i,n}$ is caused by the following operations.
\begin{enumerate}
    \item For machine $i$, let $s$ and $\sigma_{j}$ be the current configuration for $a[i]$ and $\Delta_{next_i}$. An action $\sigma_j := a[i] \leftarrow \verb"e"$ updates $a[i]$ from $s$ to $\verb"e"$, and changes $\Delta_{next_i}$ to $\sigma_{j+1}$ if $j=1 \ldots k-1$ and \verb"null" otherwise.

    \item For machine $i$, let $\hat{a}$ and $\sigma_{j}$ be the current configuration for $a[i]$ and $\Delta_{next_i}$. For all $k = 1 \ldots n, k\neq i$, let $q_k$ be the configuration of the local queue in machine $k$. An action $\sigma_j := \verb"send"(a[i])$ updates the queue from $q_k$ to $q_k\circ (a[i],\hat{a})$, and changes $\Delta_{next_i}$ to $\sigma_{j+1}$ if $j=1 \ldots k-1$ and \verb"null" otherwise.

    \item For machine $i$, for all $j= 1\ldots n$, let $s[j]$ be the current configuration for $a[j]$.
     Then an action $\sigma_j := \verb"receive"(a[i \oplus 1],\ldots, a[i \oplus n])$ performs the following updates.
     \begin{list1}
       \item This following step is done iteratively over $i'=1 \ldots n, i' \neq i$.
       \begin{list1}
              \item Let the current configuration of the local queue be $q_i = msg_1 \circ msg_2 \ldots \circ (a[i'], a_1) \circ \ldots \circ (a[i'], a_{last}) \ldots \circ msg_k$.  Let $(a[i'],a_1); \ldots; (a[i'],a_{last})$ be the subsequence of messages with variable $a[i']$ in the queue. Then $\Delta_{i,k}$ updates $Q_i$ by $msg_1 \circ msg_2 \ldots\circ \ldots msg_k$, and updates variables $a[i']$ by $a_{last}$.
                  Note if no related message exists, then variables will not be updated.
       \end{list1}
       \item It updates $\Delta_{next_i}$ from $\sigma_{j}$ to $\sigma_{j+1}$ if $j=1 \ldots k-1$ and \verb"null" otherwise.
     \end{list1}
\end{enumerate}
\end{defi}

\begin{defi} Define a GCA (global-cycle-accurate) system with n-redundancies over $\mathcal{S}_R$ as $\mathcal{S}=(\mathcal{S}_R, \Delta_\mathcal{T}, \mathcal{C})$.
\begin{list1}
    \item $\mathcal{S}_R=\bigvee_{i=1 \ldots n} \mathcal{A}_{i,n}$.
    \item $\Delta_\mathcal{T}$ is the global periodic jump with parameter $\mathcal{T}$.
    \item $\mathcal{C}$ is the global clock.
\end{list1}
\end{defi}

Intuitively, the semantics of GCA systems is that on the global
level, for each machine $\mathcal{A}_i$ the scheduling sequence is
constrained by $\mathcal{T}$; starting from time equals to zero,
for every $\mathcal{T}$ time units, it should complete the
sequence of actions defined by $\overline{\sigma}$.

\begin{defi} The change of configuration of $\mathcal{S}=(\mathcal{S}_R, \Delta_\mathcal{T}, \mathcal{C})$ is caused by following operations.
\begin{enumerate}
    \item Actions defined by $\mathcal{S}_R$.
    \item For all $x= 1 \ldots n$, $\Delta_\mathcal{T}$ reads $\bigvee_{a \in V}a[x]$, updates $V_{env_x}$, and resets $\Delta_{next_x}$ to $\sigma_1$, respectively.
    \item The clock reading of $\mathcal{C}$ changes as time advances.
\end{enumerate}
\end{defi}

\begin{defi} A GCA (global-cycle-accurate) system must satisfy the constraint: Starting from $t=0$, $\Delta_\mathcal{T}$ is always activated
with the period of $\mathcal{T}$. When $\Delta_\mathcal{T}$
terminates, configurations over $\Delta_{next_1}, \Delta_{next_2},
\ldots \Delta_{next_n}$ are always with values $\verb"null",
\verb"null",  \ldots , \verb"null"$, respectively.
\end{defi}

\subsection{Giotto} The model of computation of Giotto-like systems, e.g., Giotto, TDL,
or HTL, can thus be viewed as the case where every component of
such system is a GCA system without (1) redundancies and (2)
\verb"send" and \verb"receive" actions, where $\mathcal{T}$ can be
viewed as the logical deadline.

\section{Deterministic Assumption\label{sec:Deterministic.Assumption}}

\subsection{Deterministic Assumption in GCA Systems\label{sec:Deterministic.Assumption.In.GCA}}

Nevertheless, when applying the concept of logical execution time
to GCA systems where $n$ is not 1, intuitive extensions for
scheduling policies from Giotto make internal determinism no
longer guaranteed. The reason is that machines need to communicate
to each other to derive consensus results. However, different
scheduling policies (while the relative ordering is the same)
differ from the result. Fig.~\ref{fig:Nondeterminism} illustrates
this idea. When $M_1$, $M_2$, and $M_3$ are three machines which
implement
 the triple modular redundancy functionalities, liveness messages sent by $M_3$ will be received by $M_1$ and $M_2$. However,
 for $M_3$, due to OS scheduling decisions its execution trace can change to that of $M_4$. If so, messages might not be received successfully,
 implying that the overall system behavior differs (internal nondeterminism).

The violation of internal determinism hinders the portability of
the model. For example, the result of fault-tolerance mechanisms
may differ simply due to scheduling policies on different machines
(this is undesired). To solve this problem, we thus propose the
concept of \emph{deterministic assumption} - it is an additional
constraint where every deployment should ensure.

\begin{figure}
 \centering
 \includegraphics[width=0.4\columnwidth]{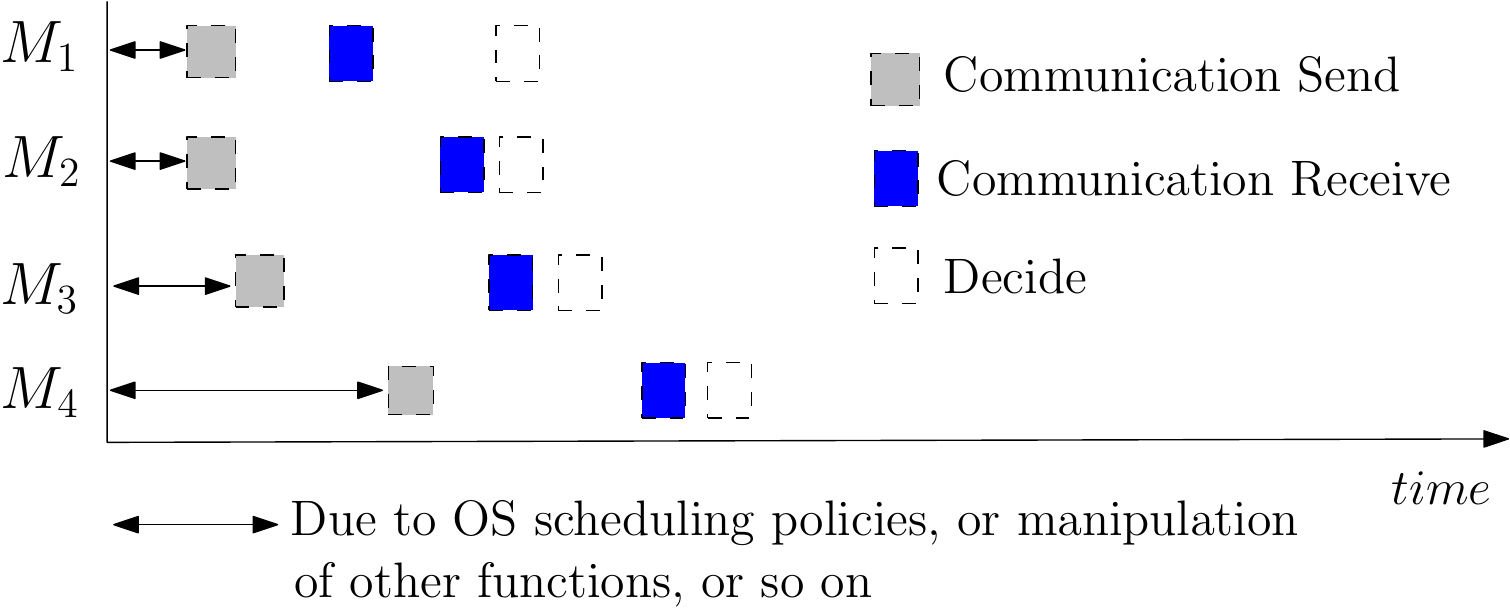}
 \caption{Internal non-determinism due to OS scheduling policies.}
 \label{fig:Nondeterminism}
\end{figure}

\begin{defi} With definition in a GCA system $\mathcal{S}$ as follows:

\begin{list1}
     \item  In the pattern definition of the machine, let $\sigma_{\alpha} := \verb"send"(a[x])$ and $\sigma_{\gamma} := \verb"send"(a[x])$
           be the predecessor and successor send operation for
          $\sigma_{\beta}:= \verb"receive"(a[x \oplus 1],\ldots, a[x \oplus n])$ in $\overline{\sigma}$, i.e.,
          $\alpha < \beta < \gamma$.

    \item On machine $i$, define $\tau_{Send,\alpha,i}$ be the clock reading from $\mathcal{C}$ where the action $\sigma_{\alpha}$ happens. If $\sigma_{\alpha}$ does not exist, then we let $\tau_{Send,\alpha,i} = -\infty$.

    \item On machine $j$, define $\tau_{Receive,\beta,j}$ be the clock reading from $\mathcal{C}$ where the action $\sigma_{\beta}$ happens.

    \item On machine $k$, define $\tau_{Send,\gamma,k}$ be the clock reading from $\mathcal{C}$ where the action $\sigma_{\gamma}$ happens.
        If $\sigma_{\gamma} := \verb"send"(a[x])$ does not exist, then we let $\tau_{Send,\gamma,k} = \infty$.

\end{list1}
then a GCA system with $n$-redundancy satisfies the deterministic
assumption if, for all machine $i$,$j$,$k$, $\tau_{Send,\alpha,i}
< \tau_{Receive,\beta,j}< \tau_{Send,\gamma,k}$.

\end{defi}

Intuitively, this means that before machine $j$ starts a receive
action for variable $a[j \oplus 1],\ldots, a[j \oplus n]$, all
messages sent to $j$ earlier with contents related to $a[j \oplus
1],\ldots, a[j \oplus n]$ should be ready in the local network
queue $Q_j$, but an unrelated message ($\sigma_{\gamma} :=
\verb"send"(a[x])$) should not arrive earlier.

\section{Impact of Deterministic Assumption\label{sec:Impacts.Deterministic.Assumption}}
In the following, we discuss how the use of deterministic
assumption influences verification and implementation.

\subsection{Deterministic Assumption Simplifies Verification}

For verification, by introducing deterministic assumption we have
explicit control over the sequential ordering. Regarding some
local properties, we can have an untimed synchronous model without
false positives.

\begin{theo} Let $\mathcal{S} = (\mathcal{S}_R, \Delta_\mathcal{T}, \mathcal{C})$  be a GCA system with n-redundancies satisfying the deterministic assumption,
and $\mathcal{S}_{sync}$ be a GCA system where each machine in
$\mathcal{S}_{sync} = (\mathcal{S}_R, \Delta_\mathcal{T},
\mathcal{C})$ executes synchronously in actions\footnote{i.e., in
each period, let $T_{i,\sigma_k}$ be the clock reading from
$\mathcal{C}$ where the $k$-th action in $\overline{\sigma}$ is
activated on machine $i$, then for all $i, j = 1 \ldots n$,
$T_{i,\sigma_k}=T_{j,\sigma_k}$.}. Then for verification
conditions $\varphi$,  $\mathcal{S} \models \varphi
\Leftrightarrow \mathcal{S}_{sync} \models \varphi$ if  $\varphi$
has the following constraints.
\begin{enumerate}
    \item Property $\varphi$ is in PLTL and does not use the operator $\mathbf{X}$.
    \item There exists an $i\in 1\ldots n$ such that for all propositional variable $p$ used in property $\varphi$, $p$ is a predicate over $V_i \cup \Delta_{next_i}$, i.e.,
           $p$ is of the format $\verb"exp = c"$,
           where $\verb"exp":= \verb"exp" + \verb"exp"\;|\;\verb"exp" * \verb"exp"\;|\;\verb"exp" / \verb"exp"\;|\;v_a\;|\;-(v_a)\;|\;\verb"const"$;
           $v_a \in V_i \cup \Delta_{next_i}$, $\verb"c"$ and $\verb"const"$ are constants (in-machine/local properties without message queue).
\end{enumerate}
\end{theo}

\begin{proof}

\noindent\textbf{(Step 0: Preparation)} To prove the theorem, we
utilize theories of projection and stutter equivalence relation
(for complete definitions, see~\cite{peled:98:theo}).
\begin{list1}
    \item For a GCA system of n-redundancies, define $\circleddash_i$ be the projection over the execution trace which preserves variables in $V_i$
    and the next action $\Delta_{next_i}$ in machine $i$
    (e.g., $(v_{i_1}, \Delta_{next_{i_1}})$$(v_{i_2}, \Delta_{next_{i_2}})...$).
    \item Let $\Sigma$ be the set of alphabets, and $\curlyvee: \Sigma^{*} \rightarrow \Sigma^{*}$ be the stutter removal operator which replaces every substring of identical alphabets by one alphabet (e.g., from $xxxyyzzz$ to $xyz$), then two words $u_1, u_2\in \Sigma^{*}$ are called \emph{stutter equivalent} if $\curlyvee(u_1) = \curlyvee(u_2)$.
\end{list1}

\noindent\textbf{(Step 1: Stuttering over projections for
asynchronous systems)}
Without loss of generality we assume $\varphi$ describes variables
$V_i$ and the next action in machine $i$. For our proof, we only
consider situations in one interval $\mathcal{T}$, since the
system performs periodically and can be proved inductively over
$\mathcal{T}$. For simplicity, in our notation, $\sigma_0$ is
equal to \verb"null".

In $\overline{\sigma}$, define a receive action
$\sigma_{k_1}:=\verb"receive"(a[x \oplus 1],\ldots, a[x \oplus
n])$ \textbf{valid} if
\begin{list1}
    \item there exists $\sigma_{k_2}:=\verb"send"(a[x])$ in $\overline{\sigma}$ where $k_2<k_1$ and,
    \item $\theta_{send,max} \geq \theta_{rec,max}$, where $\theta_{send,max}= \verb"MAX"_{\theta=0\ldots k_1-1}\{\theta|\sigma_{\theta}:=\verb"send"(a[x])\; or \; \verb"null"\}$, and $\theta_{rec,max}= \verb"MAX"_{\theta=0\ldots k_1-1}\{\theta|\sigma_{\theta}:=\verb"receive"(a[x \oplus 1],\ldots, a[x \oplus n]) \;or\; \verb"null"\}$.
\end{list1}

Without loss of generality, we assume the sequence
$\overline{\sigma}$ contains $\alpha$ $\verb"receive"()$ actions
that are valid. We use $\verb"receive"()_{j,x,y}$ to represent the
$y$-th valid $\verb"receive"()$ action, where it is the $x$-th
action in sequence $\overline{\sigma}$, executed on machine $j$.
Also we use $\sigma_{j,x}$ to represent the $x$-th action in
sequence $\overline{\sigma}$, executed on machine $j$. With
$\alpha$ valid $\verb"receive"()$ operations, we separate the
execution of $i$ into $\alpha +1$ groups, namely $
\{\sigma_{i,1};\ldots;\sigma_{i,x_1 -1}\}_1 \{
\small{\verb"receive"()}_{i,x_1,1}; \ldots; \sigma_{i,x_2 -1} \}_2
\ldots $\\$\{
\small{\verb"receive"()}_{i,x_{\alpha},\alpha};\ldots;\sigma_{i,k}
\}_{\alpha+1} $, where $\small{\verb"receive"()}_{i,x_1,1}
,\ldots,$$\small{\verb"receive"()}_{i,x_{\alpha},\alpha}$ are
valid.

Consider the influence regarding the execution of machine $j$,
whose execution sequence is the same as $i$. Note that with
deterministic assumption, the execution sequence (consider
globally) is not arbitrarily interleaved. For $\mu= 1\ldots
\alpha+1$, for the $\mu$-th group of machine $i$,
\begin{enumerate}
    \item For the valid $\verb"receive"(a[i \oplus 1],\ldots, a[i \oplus n])_{i,x_{\mu},\mu}$ action in $i$, its value is determined; the deterministic assumption assures that
          $\sigma_{j,x_{\mu}'}:=\verb"send"(a[j]), x_{\mu}'<x_{\mu}$ will be executed earlier than $\verb"receive"(a[i \oplus 1],\ldots, a[i \oplus n])_{i,x_{\mu},\mu}$, and $\sigma_{j,\tilde{x_{\mu}}}:=\verb"send"(a[j]), x_{\mu}< \tilde{x_{\mu}}$ (if existed) will be executed later than $\verb"receive"(a[i \oplus 1],\ldots, a[i \oplus n])_{i,x_{\mu},\mu}$.
    \item For an invalid $\verb"receive"()$ action in $i$, it equals to a \verb"null" operation.
    \item For other actions in $j$, their executions do not interfere $\verb"receive"(a[i \oplus 1],\ldots, a[i \oplus n])_{i,x_{\mu},\mu}$.
    Since for all other actions $\sigma_{i,\varsigma}$ of $i$, it does not retrieve the value from the queue, and $\verb"receive"(a[i \oplus 1],\ldots, a[i \oplus n])_{i,x_{\mu},\mu}$, the first action in this group, is determined, it is determined based on the previous action $\sigma_{i,\varsigma-1}$, regardless of actions in $j$.
\end{enumerate}

Results from (1) to (3) imply that the execution of actions on $j$
leads to stuttering states when the operator $\circleddash_i$ is
actuating over the state space.

\noindent\textbf{(Step 2: Stuttering-projective equivalence)}
Let $c_a$ and $c_b$ be the system configuration over the
projection of $\circleddash_i$. Let $c_a
\rightarrow_{\{\sigma_{k_i}\}} c_b$ represent the configuration
change from $c_a$ to $c_b$ caused by actions $\sigma_k$ in machine
$i$, and $c_a \rightsquigarrow_i c_b$ represent the configuration
change from $c_a$ to $c_b$ caused by actions operating on deployed
machines except $i$, and the advancing of the clock. For
$\mathcal{S}$, the projection of the execution trace over machine
$i$ with operator $\circleddash_i$ in $\mathcal{T}$ is
\[ c_{i_1} \rightsquigarrow_i c_{i_1} \rightarrow_{\{\sigma_{1,i}\}} c_{i_2} \rightsquigarrow_i c_{i_2} \rightarrow_{\{\sigma_{2,i}\}} c_{i_3} \rightsquigarrow_i c_{i_3}  \ldots c_{i_{k+1}}\]

Consider $\mathcal{S}_{sync}$ where for all machine $i= 1\ldots
n$, they execute synchronously in action level. The synchronous
execution implies that in every $\mathcal{T}$, each machine
executes concurrently first $\sigma_1$, then $\sigma_2$, and so
on. By applying the analysis similar to step 1, the projection of
the execution trace for $\mathcal{S}_{sync}$ over machine $i$ with
operator $\circleddash_i$ in $\mathcal{T}$ is
\[ c_{i_1} \rightarrow_{\{\sigma_{1,1}, \ldots,\sigma_{1,n}\}} c_{i_2} \rightarrow_{\{\sigma_{2,1}, \ldots,\sigma_{2,n}\}} c_{i_3}  \ldots c_{i_{k+1}}\]

We thus derive the stuttering equivalence between projective
traces of $\mathcal{S}_{sync}$ and those of $\mathcal{S}$ over
machine $\circleddash_i$.

\noindent\textbf{(Step 3: Proof)} We conclude the theorem based on
the following statements.
\begin{list1}
 \item The specification  $\varphi$ is an in-machine property, describing the state evolvement of machine $i$ without terms related to message queues.
 \item The specification $\varphi$ without the next operator $\mathbf{X}$, based on \cite{peled:98:theo}, is stutter-closed, meaning that $\varphi$ can not distinguish two stutter-equivalent sequences.
 \item For the deployed system $\mathcal{S}$, the projective trace of $\mathcal{S}$ and the projective trace
 of $\mathcal{S}_{sync}$ are stutter equivalent over $\circleddash_i$.
            Thus $\mathcal{S} \models \varphi \Leftrightarrow \mathcal{S}_{sync} \models \varphi$, which completes the proof. 
\end{list1}

\end{proof}

The result of the theorem can be discussed in four directions.
\begin{list1}

\item \textbf{(Modeling)} Previously, it is difficult to model the
behavior of parallel machines which are asynchronous in action
(micro-instructions) level but synchronous in the logical level.
The result of the theorem brings significant ease for the
verification model construction. Our previous attempt is to model
the system using the model checker SPIN \cite{holzmann:2004:smc},
which enables the modeling of asynchronous behavior. However, the
channel definition only allows the synchronization of two
machines, making the synchronous modeling difficult.

\item \textbf{(Platform independent result)} Furthermore, one
immediate result from theorem $1$ is that deterministic assumption
enables us to "prove GCA system once, result valid for all". For
any deployed systems, once they satisfy the deterministic
assumption, properties always hold regardless of the actual
execution time.

\item \textbf{(Required time for verification)} Regarding the run
time of verification, our technique makes the verification of such
systems practicable. Originally the GCA model is an asynchronous
model; with the theorem we can construct a property-preserving
verification model by synchronizing actions. For this verification
model, the size of reachable state space traversed in verification
is exponentially smaller compared to the original GCA model.

\item \textbf{(Practicability of local properties)}  At first
glance, the use of local properties seems to restrict the
applicability. Nevertheless, since each machine can keep values
sent from other machines, and fault-tolerance mechanisms mostly
manipulate over these variables, our observation is that local
properties constrained by the theorem are powerful enough for
practical use (for examining fault-tolerance mechanisms).
\end{list1}

\subsection{Deterministic Assumption Modifies Scheduling Policies}

\subsubsection{Transition and transmission need time} In practice, actions and network transmission take time.
Nevertheless, the application and theory still apply. Here we give
an guidance regarding how deterministic assumption should be
assumed in practice.

\begin{defi} With definition in an \textbf{implementation} of a GCA system $\mathcal{S}$ as follows.
\begin{list1}
     \item Let an action of the type $\verb"send"()$ only broadcasts the message to the network without updating the queue of other machines.

     \item In the pattern definition of the machine, let $\sigma_{\alpha} := \verb"send"(a[x])$ and $\sigma_{\gamma} := \verb"send"(a[x])$
           be the predecessor and successor send operation for
          $\sigma_{\beta}:= \verb"receive"(a[x \oplus 1],\ldots, a[x \oplus n])$ in $\overline{\sigma}$, i.e.,
          $\alpha < \beta < \gamma$.

    \item On machine $i$, define $\tau_{Send,End,\alpha,i}$ be the time for machine $i$ where the action $\sigma_{\alpha}$ finishes
          sending the message to the network.

    \item On machine $j$, define $\tau_{Receive,Start,\beta,j}$ be the clock reading from $\mathcal{C}$ where the action $\sigma_{\beta}$ starts.

    \item On machine $j$, define $\tau_{Receive,End,\beta,j}$ be the clock reading from $\mathcal{C}$ where the action $\sigma_{\beta}$ ends.

    \item On machine $k$, define $\tau_{Send,Start,\beta,k}$ be the clock reading from $\mathcal{C}$ where the action $\sigma_{\gamma}$ starts.

    \item Let $T_{i,j}$ be the estimation of the worst case message transmission
    time between a message is sent from machine $i$ and present in the message queue of machine $j$, and define
    $\tau_{net} = \verb"max"_{i,j \in 1\ldots n}T_{i,j}$.
\end{list1}
then an implementation of a GCA system with $n$-redundancy
satisfies the deterministic assumption if, for all machine
$i$,$j$,$k$, $\tau_{Send,End,\alpha,i}+\tau_{net} <
\tau_{Receive,Start,\beta,j}$ and $\tau_{Receive,End,\beta,j} <
\tau_{Send,\gamma,k}$.

\end{defi}

\subsubsection{Least constraint for scheduling}
Previously, in Giotto the scheduling only needs to ensure that all
actions (micro-instructions in Giotto) executed in a major-tick
will not exceed the period of task $\mathcal{T}$. However, with
the deterministic assumption the implemented scheduling should be
modified. Nevertheless, we argue that the additional constraint
induced by deterministic assumption is the least, since it only
captures the dependencies over actions.

Also, modification of the scheduling policies can be minor. For
example, if the underlying execution platform is a synchronous
system (e.g., Esterel, VHDL), then no modification is needed. If
the underlying execution platform is a asynchronous system with
RTOS, if in $\overline{\sigma}$ no duplicate send actions occur,
then one method can be applied for assuring the deterministic
assumption: a fixed window (with length $\tau_{net}$) can be set
to distinguish between all sends (in the left) and receives (in
the right), shown in fig.~\ref{fig:Time.Window}.

At the same time, for a particular implementation to examine
whether it satisfies the deterministic assumption can be checked
using verification engines for timed systems, for example, UPPAAL
\cite{larsen:1997:un} or HyTech \cite{henzinger:1997:hmc}. Note
that it is merely the protocol (irreverent of data) that needs to
be checked, which greatly releases the burden of those
engines\footnote{Complexity for the verification of timed systems
is high, where we found it not suitable to have a model with data
information be checked.}.

\begin{figure}
 \centering
 \includegraphics[width=0.45\columnwidth]{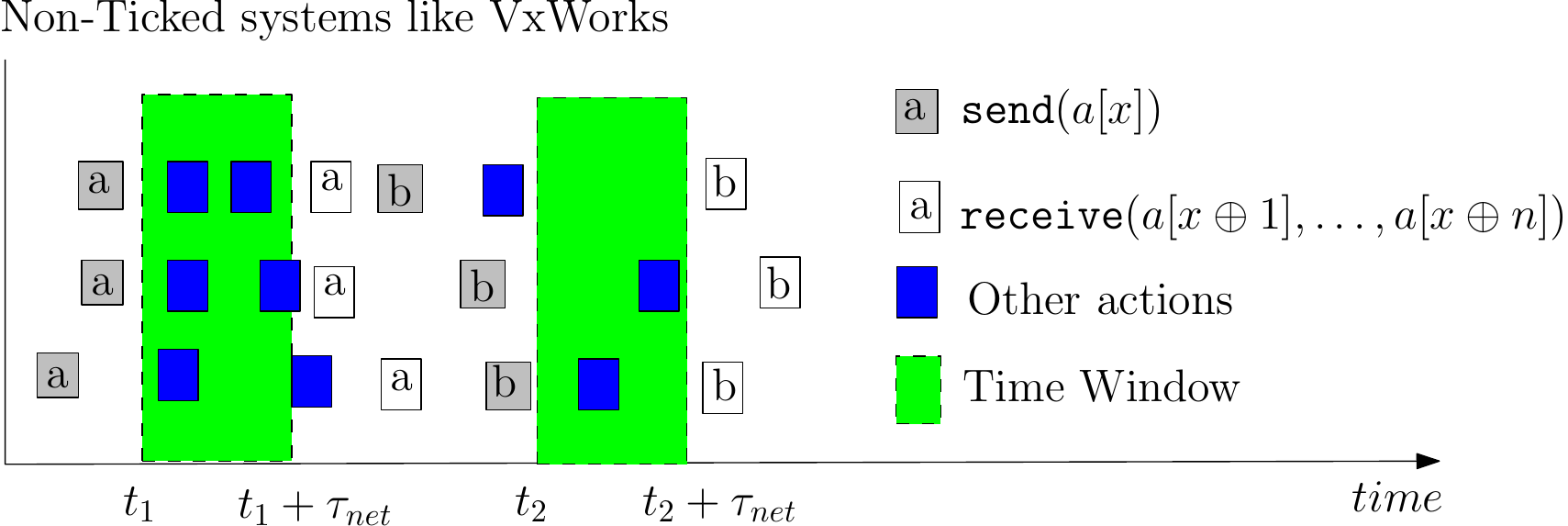}
 \caption{An example to use time window for ensuring the deterministic assumption.}
 \label{fig:Time.Window}
\end{figure}

\section{Effect of Faults\label{sec:Effect.Fault}}

In this section, we formally define the fault model used in FTOS,
and discuss the effect when faults are actuating on machines.

\subsection{Formal Construction of Fault Models\label{sec:Fault.Models}}

\begin{defi} Define the set of possible faults over a GCA system $\mathcal{S}$ be $\bigcup_{i= 1\ldots m}\xi_i$, where each fault $\xi$ is a $8$
tuple $(\verb"act", \verb"type", \sigma, i, j, k, k', \psi_k)$.
\begin{list1}
    \item  $\verb"act"$ is a boolean variable.
    \item  $\verb"type" \in\{\verb"WrongResult", \verb"FailSilent",$$\verb"MessageLoss", \verb"Corruption", \verb"Masquerade"\}$\footnote{The set of types was influenced by the IEC-61508 standard concerning networking faults. WrongResult is used to describe internal computing faults. The fault of MessageDelay will be later described and MessageRepetition as well as
             IncorrectSequence are of no influence in our MoC.
} is the type of fault.
    \item  $\sigma \in \verb"atomic"(\overline{\sigma})$.
    \item  $i\in \{1,\ldots, n\},\;j \in \{1,\ldots, |\overline{\sigma}|\}$ are fault-actuating indexes.
    \item  $k, k' \in \{1\ldots n\}$ are parameters.
    \item  $\psi_k: \verb"range"(\sigma)\rightarrow \verb"range"(\sigma)$ is the error function,
            where $\verb"range"(\sigma)$ is the range of $\sigma$.
\end{list1}
\end{defi}

The effect of faults is summarized as follows. Without loss of
generality assume every $\xi$ mentioned is actuating on machine
$i$.
\begin{list1}
    \item Let $\xi = ($\verb"act"$, \verb"WrongResult", \small{\sigma:= a[m]\leftarrow \verb"e"}\normalsize,
    i, j, k, k', \psi)$, then
        \begin{list1}
            \item At time $t$, if the value of \verb"act" is true, and $\Delta_{i,j}:= a[m]\leftarrow \verb"e"$ is activated,
            then for machine $i$, $\psi_k \circ \sigma$ updates the value of $a[m]$ to $\psi(e)$.
        \end{list1}
    \item Let $\xi = ($\verb"act"$, \verb"FailSilent", \sigma, i, j, k, k', \psi)$, then
            \begin{list1}
            \item At time $t$, if the value of \verb"act" is true, and $\Delta_{i,j}:= \sigma$ is activated, let
            $v_{i}$ be the previous configuration for $V_{i}$ before actuating $\sigma$, and $q_1, q_2, \ldots q_n$ be the
            configuration of message queues $Q_1, Q_2, \ldots Q_n$. Then in effect $\psi_k \circ \sigma$ does not update any value, i.e.,
            $V_{i}$ from $v_{i}$ to $v_{i}$, and for all $Q_i, i= 1 \ldots n$, from $q_i$ to $q_i$.
        \end{list1}
    \item Let $\xi = ($\verb"act"$, \verb"MessageLoss", \verb"send"(a[i]), i, j,$$ k, k', \psi)$, then
         \begin{list1}
            \item At time $t$, if the value of \verb"act" is true, and $\Delta_{i,j}:= \verb"send"(a[i])$ is activated, let $a$ be the previous
            configuration for $a[i]$. Let
            $q_1, q_2, \ldots q_n$ be the configuration of message queues $Q_1, Q_2, \ldots Q_n$. Then in effect
            $\psi_k \circ \sigma$ updates $q_{s}$, $s=1\ldots n,\; s\neq i,\; s\neq k$ to $q_s\circ (a[i], a)$, but
            $q_k$ to $q_k$.
        \end{list1}
    \item Let $\xi = ($\verb"act"$, \verb"Corruption", \verb"send"(a[i]), i, j,$$ k, k', \psi)$, then
         \begin{list1}
            \item At time $t$, if the value of \verb"act" is true, and $\Delta_{i,j}:= \verb"send"(a[i])$ is activated. Let $a$ be the previous
            configuration for $a[i]$. Let $q_1, q_2, \ldots q_n$ be the configuration of message queues $Q_1, Q_2, \ldots Q_n$. Then in
            effect $\psi_k \circ \sigma$ updates $q_{s}$, $s=1\ldots n,\; s\neq i,\; s\neq k$ to $q_s\circ (a[i], a)$, but $q_k$ to $q_k\circ(a[i], k')$.
        \end{list1}
    \item Let $\xi = ($\verb"act"$, \verb"Masquerade", \verb"send"(a[i]),i, j,$$ k, k', \psi)$, then
         \begin{list1}
            \item At time $t$, if the value of \verb"act" is true, and $\Delta_{i,j}:= \verb"send"(a[i])$ is activated. Let $a$ be the previous
            configuration for $a[i]$. Let
            $q_1, q_2, \ldots q_n$ be the configuration of message queues $Q_1, Q_2, \ldots Q_n$. Then in effect $\psi_k \circ \sigma$ updates  $q_{s}$, $s=1\ldots n,\; s\neq i,\; s\neq k'$ to $q_s\circ (a[i], a)$,
            but $q_{k'}$ to $q_{k'}\circ(a[k], a)$.
        \end{list1}
\end{list1}

In general, all faults can be viewed as being controlled by a
fault automaton.

\begin{defi} A fault model over a GCA system $\mathcal{S}$ is a extended timed automaton $\mathcal{A}_{fault}=(L, AP, T,  I, E, \Sigma, \bigcup_{i= 1\ldots m}\xi_i)$.
\begin{list1}
\item  $L$ is set of locations. \item
$AP=\{\verb"act"_1,...\verb"act"_m\}$ is a finite set of atomic
propositions, where $\verb"act"_i \in \mathbb{B}$ represents
       the activation status of a particular fault type.
\item  $T$ is a finite set of clocks. \item  $I$ is the invariant
condition mapping elements in $L$ to clock constraints. \item  $E$
is the location switch. \item  $\Sigma$ is a mapping $L
\rightarrow \mathbf{2}^{AP}$ indicating the set of faults
activating in the control location. \item  $\bigcup_{i= 1\ldots
m}\xi_i$ is the set of possible faults.
\end{list1}
\end{defi}

\subsection{Properties\label{sec:Properties}}

During execution, the valuation of
$\{\verb"act"_1,...\verb"act"_m\}$ can be viewed as a timed run
$(\overline{s}, \overline{\nu})=(s_1, t_1)(s_2, t_2)\ldots$
defined by the fault model automaton, where $s_1, s_2,\ldots \in
\mathbb{B}^{m}$ are valuations of atomic propositions, and $t_1,
t_2, \ldots \in \mathbb{R}$ are durations. For the complete
definition of timed run, see \cite{alur:1994:tta}. Let
$\xi=(\verb"act"_{\alpha}, \verb"type", \sigma_{\alpha}, i, j, k,
k', \psi)$. At time $t$, if $\verb"act"_{\alpha}$ is evaluated to
be true by $\mathcal{A}_{fault}$, and $\sigma_{\alpha}$ is
executing, then the fault happens.

In order to give the result where the theorem still holds with
fault models, we introduce a series of definitions used as
assumptions.

\begin{defi} Let $\mathcal{A}_{fault}=(L, AP, T,  I, E, \Sigma, $$\bigcup_{i= 1\ldots m}\xi_i)$ be the fault automaton actuating over the GCA system $\mathcal{S}$.
First define an infinite sequence
$(\overline{S_T},\overline{T_T})=(S_1,t_1)(S_2,t_2)\ldots$ as the
transition triggering sequence of $\mathcal{S}$ where
\begin{list1}
    \item For an integer $i \geq 1$, $S_i$ is the set of actions actuating on time $t_i$.
    \item Sequence $t_1,t_2,t_3 \ldots$ is a strictly ascending time sequence.
\end{list1}
Given a timed run $(\overline{s}, \overline{\nu})$ over
$\{\verb"act"_1,...\verb"act"_m\}$ induced by
$\mathcal{A}_{fault}$, and a transition triggering sequence
 $(\overline{S_T},\overline{T_T})$ of $S$, define the actuating fault sequence $\zeta_{(\overline{s}, \overline{\nu})}$ be an infinite sequence $(s_{\overline{1}}, t_{\overline{1}})(s_{\overline{2}}, t_{\overline{2}})(s_{\overline{3}}, t_{\overline{3}})\ldots$ where
\begin{list1}
    \item For all $\overline{i} \geq 1$, $s_{\overline{i}}$ is the set of actions influenced by $(\overline{s}, \overline{\nu})$.
    \item $t_{\overline{1}},t_{\overline{2}},t_{\overline{3}} \ldots$ is a strictly ascending time sequence.
\end{list1}
\end{defi}

\begin{defi}
 Let \small{$(\overline{S_{T_1}},\overline{T_{T_1}})=(S_{11},t_{11})(S_{21},t_{21})\ldots$} \normalsize and \small{$(\overline{S_{T_2}},\overline{T_{T_2}})=(S_{12},t_{12})(S_{22},t_{22})\ldots$} \normalsize be two transition triggering sequences over a GCA system $S$. Define two actuating fault sequences $\zeta_{(\overline{s}_1, \overline{\nu}_1)}$ over $(\overline{S_{T_1}},\overline{T_{T_1}})$ and $\zeta_{(\overline{s}_2, \overline{\nu}_2)}$ over $(\overline{S_{T_2}},\overline{T_{T_2}})$ be
effect-indistinguishable over $\mathcal{S}$ if the following
holds.
\begin{list1}
   \item Starting from time equals $0$, for each interval with length $\mathcal{T}$ (e.g., $[0,\mathcal{T}),[\mathcal{T},2\mathcal{T})\ldots$), let $(S_{1\alpha_1},t_{1\alpha_1})(S_{1\alpha_2},t_{1\alpha_2})\ldots\\(S_{1\alpha_k},t_{1\alpha_k})$ and
        $(S_{2\beta_1},t_{2\beta_1})(S_{2\beta_2},t_{2\beta_2})...(S_{2\beta_j},t_{2\beta_j})$ be the
        subsequences of $\zeta_{(\overline{s}_1, \overline{\nu}_1)}$ and $\zeta_{(\overline{s}_2, \overline{\nu}_2)}$ contained in this interval. Then
        $k=j$, and $\forall m= 1 \ldots k$, $S_{1\alpha_m}=S_{2\beta_m}$, i.e., two untimed actuating fault subsequences in the interval are identical.
\end{list1}
\end{defi}

\begin{defi}
Let $\mathcal{F}_1$, $\mathcal{F}_2$ be the fault model actuating
on a GCA system $\mathcal{S}$. Define  $\mathcal{F}_1$ and
$\mathcal{F}_2$ effect-indistinguishable if for all timed run
$(\overline{s_1}, \overline{\nu_1})$ of $\mathcal{F}_1$, for its
corresponding actuating fault sequence $\zeta_{(\overline{s}_1,
\overline{\nu}_1)}$, there exists a timed run $(\overline{s_2},
\overline{\nu_2})$ of $\mathcal{F}_2$ with corresponding actuating
fault sequence
 $\zeta_{(\overline{s}_2, \overline{\nu}_2)}$ such that $\zeta_{(\overline{s}_1, \overline{\nu}_1)}$ and $\zeta_{(\overline{s}_2, \overline{\nu}_2)}$
are effect-indistinguishable over $\mathcal{S}$, and vice versa.
\end{defi}

\begin{theo} Let $\mathcal{S}$ be a GCA system with n-redundancies satisfying the deterministic assumption,
and $\mathcal{S}_{sync}$ be a GCA (global-cycle-accurate) system
where each machine in $\mathcal{S}_{sync}$ executes synchronously
in actions. Let $\mathcal{F}$, $\mathcal{F}_{sync}$ be the fault
model actuating on $\mathcal{S}$ and $\mathcal{S}_{sync}$,
respectively. If $\mathcal{F}$ and $\mathcal{F}_{sync}$ are
effect-indistinguishable, then for verification conditions
$\varphi$,  $\mathcal{F}\times\mathcal{S} \models \varphi
\Leftrightarrow \mathcal{F}_{sync}\times\mathcal{S}_{sync} \models
\varphi$.

if  $\varphi$ has the following constraints.
\begin{enumerate}
    \item Property $\varphi$ is in PLTL and does not use the operator $\mathbf{X}$.
    \item There exists an $i\in 1\ldots n$ such that for all propositional variable
       $p$ used in property $\varphi$, where $p$ is a predicate over $V_i  \cup \Delta_{next_i}$.
\end{enumerate}
\end{theo}

\begin{proof}
The theorem follows immediately with definitions described above.
\end{proof}


\subsection{From Theory to Practice\label{sec:From.Theory.Practice}}

\subsubsection{Overapproximation of the fault model}
Based on previous sections, we can construct the system model
which is synchronous with property preservation. However, it may
not be easy (or sometimes impossible) to construct a fault model
which is effect-indistinguishable as synchronous models. In this
way, an over-approximation of the fault model (where time is
abstracted to ticks) is sometimes necessary.


In FTOS, the occurrence between faults is measured by the
\emph{least time between failure} (LTBF). Let the system period be
$\mathcal{T}$ and the LTBF be $\eta$, then a extremely safe
approximation for number of occurrence of faults during a period
$\mathcal{T}$ is $\lceil\frac{\mathcal{T}}{\eta}\rceil$. Since the
above approximation can be too strong, a more desirable way to
underapproximate the LTBF can be as follows:  when
$\eta>3\mathcal{T}$, the time for the occurrence of two
consecutive faults should contain at least
$\lfloor\frac{\eta}{\mathcal{T}}\rfloor -1$ complete
periods\footnote{The untimed synchronous model can still specify
time, but with only finite granularity - the minimum time unit is
decided by $\mathcal{T}$ specified in the GCA model.}. When $\eta$
is large, the set of behaviors defined in the approximated fault
model is close to that of the original fault model. Since in
practice the fault model is based on probability measures, the
impact of such over-approximation of fault models is minor,
meaning that \emph{a spurious counter-example is still a trace
that the system can not defend}.

With the above descriptions, an untimed fault model compatible
with the verification model of GCA can be
constructed\footnote{Still, the LTBF can be vary large for
modeling. For this problem we have established some theoretical
criteria to reduce the LTBF efficiently with property preservation
\cite{Cheng:2009:Timing}.}.

\subsubsection{Scheduling policies in actual deployment}
In actual deployment (implementation), the set of traces may be
constrained compared to the original GCA model due to the
scheduling policy of the local machine. The reminder is that with
the fault effect \verb"masquerade", the synchronous model can
demonstrate more possibilities than a deployment because the
ordering for the correct and masqueraded messages in the
synchronous model is not uniquely determined, whilst an actual
deployment might have the order fixed due to the restricted
behavior.

\subsubsection{Message delay}
In our formal model construction, the message delay is not
considered for simplicity reasons. Nevertheless, it is possible to
extend the fault model with a timed queue for the storage of
delayed messages. In practice, the error of message delay can be
detected by our tool with an over-approximated fault model.

\section{Implementation: FTOS-Verify\label{sec:Implementation}}

We have implemented our first version of
\emph{FTOS-Verify}\footnote{The software can be retrieved on
demand.}, which is deployed as an Eclipse add-on for FTOS. In this
section, we summarize the current features.
\begin{enumerate}
    \item It enables an automatic translation from FTOS models into verification models; choices of analysis techniques
can be made between fault propagation and interval analysis.
    \item It automatically generates some formal properties regarding fault-tolerance.
    \item Also we elaborate on integrating existing model checkers into the Eclipse framework; under the Eclipse IDE the verification model
can be checked, where verification results are reorganized and
shown under the Eclipse console.
    \item Finally, since some counter example (in the trace file) may reach 300000 lines, we perform automatic analysis to prune out all
          unnecessary details such that the counter example is easy for non-verification experts to understand (around 300 to 700 lines).
\end{enumerate}

\subsection{Template-based Model Generation}
In previous sections, we outline all actions, while the
implementation details for each action (mechanism) are unknown. In
reality, we extract all mechanisms from existing
platform-dependent templates, and rewrite them into formats
acceptable by model checker. Since for some implementation
platforms, an action corresponds to a sequence of code which in
effect performs the update, we hope to maintain the "shape" of
sequential execution while the model is executed synchronously -
by doing so, these verification templates describing
fault-tolerance mechanisms can be easily adapted to generate new
templates for other platforms. To achieve this goal, we use
Cadence SMV \cite{mcmillan:cs} as our back-end verification
engine, for the reason that it supports an extended description
format which resembles sequential execution.

In FTOS, we apply template-based code generation to describe each
action based on the meta-model layer. In this way, for an user
defined model with different names (e.g., ports), when a model is
generated, instructions and formal specifications will be
annotated with correct names accordingly.

\subsection{Analysis Techniques 1: Boolean Program for Fault Propagation}
Currently verification of software systems relies heavily on
abstraction techniques. For example, in the SLAM
\cite{ball:2001:st} project the program is converted into boolean
programs based on predicate abstraction techniques. In our case, a
simple abstraction technique is to apply \emph{fault abstraction},
where we define a value of a port to be either \verb"Correct" or
\verb"Erroneous". These two values reflect the objective status of
the value in the port. Combined with the platform-independent
fault-tolerant mechanism, this method can provide good indication
regarding the correctness of fault tolerance mechanisms. In
FTOS-Verify, it is applicable for two built-in tests, namely
\verb"TestPortAbsolute" and \verb"TestLiveness".

At the same time, corresponding abstraction should be established
for the specification. Here the specification is generated
manually once accompanied with the code template for
fault-tolerance mechanisms, and it can be continuously reused
after being created.

\subsection{Analysis Techniques 2: Preliminary Bounded Interval Analysis}

In FTOS-Verify, we also experiment bounded interval analysis for
examining the validity of the test \verb"TestPortRelative", where
the user may specify the value domain for each port and input, and
a richer set of specification can thus be verified.

\section{Case Study: Balanced-Rod System \label{sec:Case.Study}}

By applying formal verification, we can test the applicability of
new or built-in mechanisms under different fault models before
deployment. Here we illustrate a case study of a balanced-rod
system (for the hardware, see fig.~\ref{fig:HW.Model.Rod}). The
purpose of the system is to use a PID controller with some data
acquisition boards to control the rod such that it keeps balancing
while pointing upwards. We expect to have fault-tolerant abilities
in the system, and our design intension is to use triple modular
redundancy for parallel execution, applying suitable voting
mechanism to preclude malfunctioned units. We illustrate the
verification process with an unsuitable design in FTOS.

\begin{figure}
 \centering
 \includegraphics[width=0.15\columnwidth]{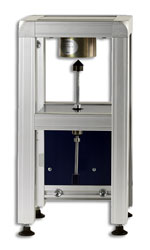}
 \caption{The demonstrator of the balanced-rod system.}
 \label{fig:HW.Model.Rod}
\end{figure}

\begin{figure}
 \centering
 \includegraphics[width=0.6\columnwidth]{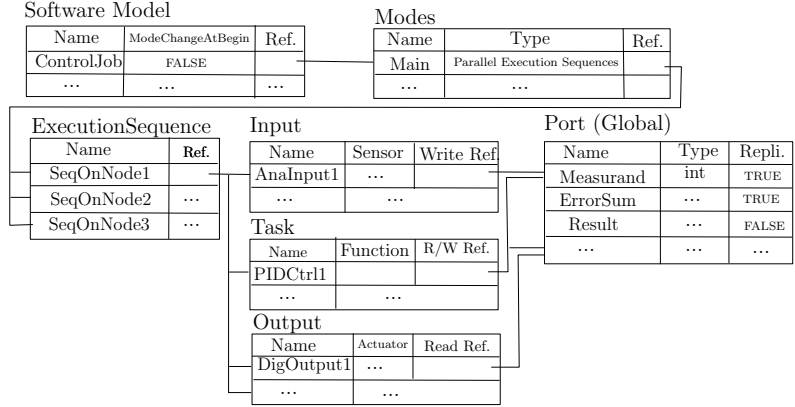}
 \caption{The software model of a improperly designed controller in FTOS (with irrelevant details omitted).}
 \label{fig:SW.Model.Rod}
\end{figure}

\subsection{Fault Model} In this system, four possible fault configurations are
possible, namely \verb"All_correct", \verb"1_2_correct",
\verb"2_3_correct",  and \verb"1_3_correct", implying that the
fault hypothesis only allows to have at most one ECU be faulty at
any instance. The fault is limited to the erroneous reading from
inputs or erroneous result of task functions.

\subsection{Software Model}  Fig.~\ref{fig:SW.Model.Rod} is the software model of the
balanced-rod controller. The job \verb"ControlJob" contains one
operation mode \verb"Main", where its execution strategy defines
three parallel sequences (\verb"SeqOnNode1", \verb"SeqOnNode2",
\verb"SeqOnNode3"). For \verb"SeqOnNode1", it consists of one
input read operation, one task, and several output operations. The
input \verb"AnaInput1" reads the data from the cache of the
dedicated hardware and stores in the port \verb"MeasureAnd". The
task function (e.g., \verb"PIDCtrl1") retrieves values from some
ports, performs stateless functions, and writes to some ports.
Each output operation (e.g., \verb"DigOutput1") retrieves the data
from one particular port and actuates. The port \verb"MeasureAnd"
has the replication property equal to \verb"TRUE", meaning that in
the deployment, the port is replicated to separate machines, but
there will be no mechanism to guarantee the consistency between
these replicated ports.

\subsection{Fault Tolerance Model}
Since only one machine should offer the output value, on the model
level a state machine \verb"REDUNDANCY_TRIGGER_TMR" is constructed
for the decision making. For example, the state machine states
that if the detected fault configuration is \verb"All_Correct" (no
machine is faulty) or \verb"1_2_Correct", then machine $1$ is
responsible for generating the output. The state machine will be
translated and deployed into each sub-system.

In this example, our fault-tolerance design first demands the
system to perform the test \verb"TestPortAbsolute" (pre-built in
FTOS) on port \verb"Result". For details, see appendix. When the
test fails, it simply performs the \verb"Ignore" operation, and
expects the mechanism of the state machine will work correctly,
ruling out the possibility for faulty machines to generate the
output.

\subsection{Properties} One local safety property automatically generated by FTOS-Verify
(e.g., annotated to machine $1$) is in fig.~\ref{fig:spec}.
Intuitively, this property specifies the case: if machine $1$
(ecu1) is responsible for offering the output\\
(\verb"REDUNDANCY_TRIGGER_TMR_operating_value = 0"), then the
value offered can not be faulty\\
(\verb"ecu1_local_ports_Out1.Result = 0").

\begin{figure}
{\begin{verbatim}Correct_DigOutput1_Result:
SPEC AG((REDUNDANCY_TRIGGER_TMR_operating_value=0)
    ->(ecu1_local_ports_DigOutput1.Result = 0));
\end{verbatim}}
 \caption{Specification automatically annotated by FTOS-Verify.}
 \label{fig:spec}
\end{figure}

\subsection{Result} In the
fault-propagation model, the result of model checking either
reports the system correct, or indicates a counter-example showing
a trace of fault emersion and propagation which leads to errors.
In this example, the error trace indicates the situation when
faults happen alternately between two fault configuration units.
The scenario (reported by model checkers; here simplified for
explanation) is as follows:

\begin{figure}
 \centering
 \includegraphics[width=0.5\columnwidth]{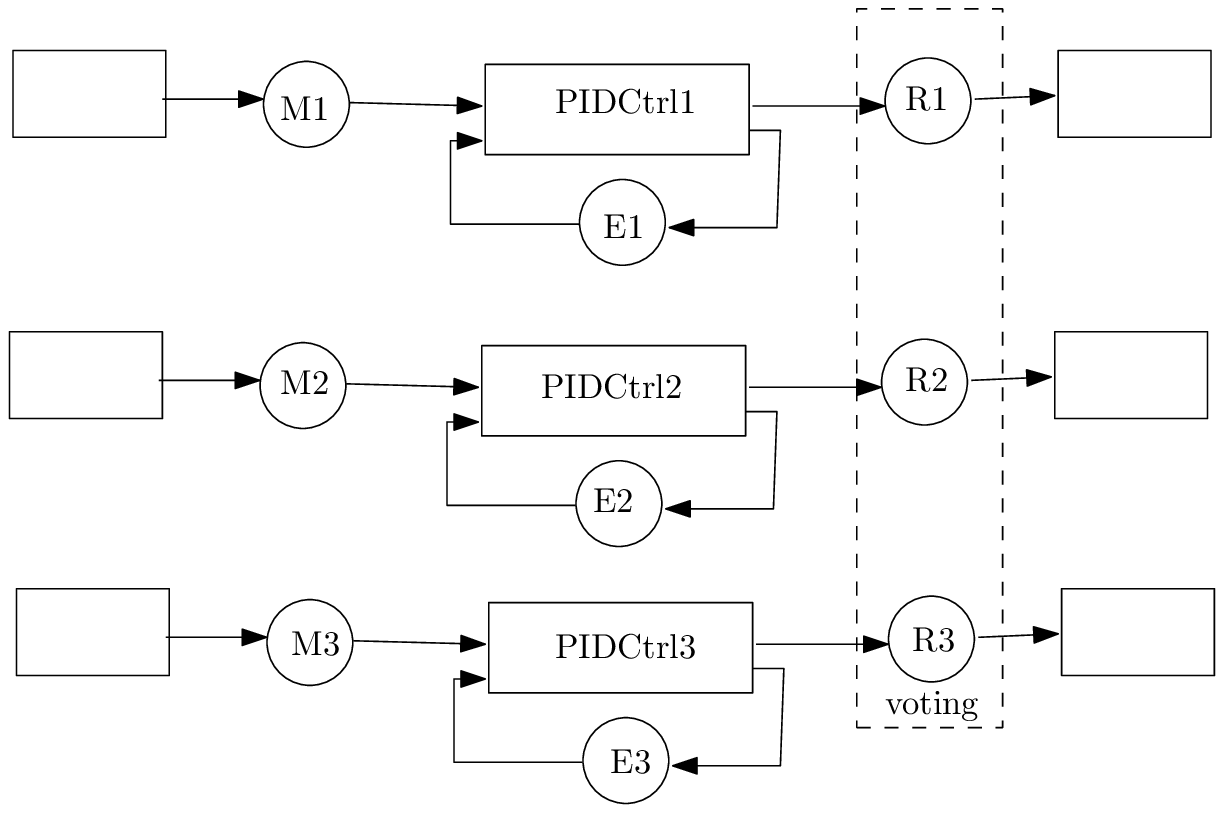}
 \caption{The deployment model of a wrongly designed controller (with irrelevant details omitted).}
 \label{fig:Fault.Propagation.Example}
\end{figure}

\begin{list1}
\item In early stages of the first period, input \verb"I1"
generates faulty results and stores into\\
\verb"ecu1_global_ports.MeasureAnd". As \verb"PIDCtrl1" uses data
from \verb"ecu1_global_ports.MeasureAnd", the result generated by
\verb"PIDCtrl1" is faulty, implying
\verb"ecu1_global_ports.ErrorSum" and\\
\verb"ecu1_global_ports.Result" being faulty.

\item In later stages of the first period, values of
\verb"ecu1_global_ports.Result", \verb"ecu2_global_ports.Result",
and \verb"ecu3_global_ports.Result" will be sent to other machines
and examined using voting mechanism. For ecu1, values from ecu2
and ecu3 are stored in \verb"ecu1_port_from2.Result" and
\verb"ecu1_port_from3.Result", respectively. All results of
\verb"TestPortAbsolute" from each machine indicate that ecu1 is
faulty, and the\\ \verb"REDUNDANCY_TRIGGER_TMR" on each machine
indicates that output should be generated by ecu2. Also, the error
happened in machine $1$ is ignored.

\item In early stages of the second period, input \verb"I2"
generates faulty results and stores into\\
\verb"ecu2_global_ports.MeasureAnd". Therefore, the result
generated by \verb"PIDCtrl2" is faulty, implying\\
\verb"ecu2_global_ports.Result" and \verb"ecu2_global_ports.Error"
being faulty. Although \verb"I1" is not faulty, \verb"PIDCtrl1"
still cannot generate correct result because it also relies on the
value in \verb"ecu1_global_ports.ErrorSum". Therefore,
\verb"ecu1_global_ports.Result" is still faulty at this period.

\item In later stages of the second period, values of
\verb"ecu1_global_ports.Result", \verb"ecu2_global_ports.Result",
and \verb"ecu3_global_ports.Result" will be again sent and
examined. As values of \verb"ecu1_global_ports.Result" and
\verb"ecu2_global_ports.Result" are faulty\footnote{This can be
further interpreted as the case where ecu1 and ecu2 hold the same
erroneous value.}, machine $1$ and $2$ will treat $3$ faulty,
whilst machine $3$ treats $1$ and $2$ faulty later when exchanging
information regarding the correct status of machines, opinions of
machine $3$ is precluded, and \verb"REDUNDANCY_TRIGGER_TMR"
decides that machine $1$ is responsible for offering the output,
which is undesired (although the system follows desired fault
configurations).

\item A fix of the example is to add port unification strategies
over port \verb"ErrorSum" with communication after\\
\verb"TestPortAbsolute"(). For example, in ecu1, take the median
value among \{\verb"ecu1_global_ports.ErrorSum",
\verb"ecu1_port_from2.ErrorSum", \verb"ecu1_port_from3.ErrorSum"\}
and store it as an update for\\ \verb"ecu1_global_ports.ErrorSum".

\end{list1}

\subsection{More testcases and counter examples}The document of the software also describes cases with much complicated counter-examples for proving the
invalidity of some fault-tolerant mechanisms under certain fault
models.

\subsection{Efficiency} The time required to generate all models is approximately $4$
seconds\footnote{Programs are executed under a machine with Intel
Dure Core 2.6 Ghz CPU and 2GB memory.}, and the size of the
generated model can vary between $7000$ to $11000$ lines. The
process of verification and generation of the counter example
takes nearly $60$ seconds for the simplest case - time increases
when complicated mechanisms are equipped in the model. For similar
cases where the property is proven to be correct, the required
time is no greater than $250$ seconds. An earlier (preliminary)
version of FTOS-Verify based on SPIN without using the
deterministic assumption was not able to generate results in
reasonable time ($<60$ min) even for the simplest testcase. These
numbers show a drastic reduction of the required time and
demonstrate the effect of the deterministic assumption.

\section{Related Work\label{sec:Related.Work}}

Existing work on design or verification for fault-tolerance
mechanisms can be categorized into two different categories.
Within the first category, researchers are focusing on verifying
the applicability of a single fault-tolerance mechanism based on a
concrete fault-model. For the second category researchers are
offering languages or methodologies for the use of verification,
e.g., \cite{bernardeschi:2000:fvf,owre:1995:fvf}.
Nevertheless, the above work does not place their focus on
automatic model generation for verification. As the system models
and the fault models influence the applicability of the mechanism,
when they are modified, the corresponding verification models are
required to be reconstructed. Construction and modification
manually can be time consuming or error prone; in FTOS-Verify, the
verification model is generated automatically as the model
changes.

FTOS is not the only system which focus on model based development
for fault-tolerant systems. Pinello et al
\cite{Pinello:2004:FTDF}. also focus on model-level description of
software models, hardware models, and fault models, and allow
automatic synthesis from models to deployments; their model (FTDF)
is based on the extension of the dataflow model. Formal analysis
techniques based on fault-propagation are later proposed
\cite{mckelvinjr:2005:faf}. There are inherent differences between
two MoCs which make it difficult to compare the results. Our
observation is that the tree-construction algorithm (similar to
backward reachability analysis) as proposed in FTDF does not
really apply symbolic techniques commonly used in verification and
might suffer from exponential complexities\footnote{The
complexities is $\mathbf{O}(D^{2M})$, where $D$ is the average
number of inputs in an actor, and $M$ is the number of actors,
given an undesired event.}; we relate ourselves to the
verification engine with symbolic state space manipulation, making
the detection of counter examples in large systems possible.
Secondly, our analysis is not restricted to reachability analysis
but enables the use of temporal logic. Lastly, our theorem enables
the state-space reduction. A conjecture is that our theorem can be
further extended to give a supporting theorem for their
MoC\footnote{i.e., there exists a padding of no-ops for a FTDF
model such that local properties can be verified synchronously.}.

There are other model based tools for embedded control with
verification abilities integrated, e.g., \cite{schaefer:2008:crm},
but verifying fault-tolerance mechanisms is not their primary
focus. Also we find in most cases, the deployment from the model
is synchronous, where FTOS is focusing on the deployment over
either synchronous or asynchronous systems.

\section{Conclusion and Future Work\label{sec:Conclusion.Future.Work}}
Our contribution is summarized as follows.
\begin{list1}
    \item We concentrate on the modeling and verification of fault-tolerant systems, and mathematically construct the GCA system, the
          formal model of FTOS, for the use of verification.
    \item The deterministic assumption relates all implementations with common features regardless of the underlying MoC
          of the platform (synchronous or asynchronous). It also simplifies the model
          construction, and reduces the reachable state space for verification. Due to the redundancies in fault-tolerance systems,
          constrained local properties used in verification can resemble global specifications.
          The work can be also applied for the analysis of systems based on the FLET paradigms proposed within Giotto.
   \item Regarding implementation, combined with existing templates in the synchronous platform and the asynchronous platform, the presented verification technique enables designers with limited knowledge in verification
          to test the applicability of fault-tolerance mechanisms formally.

\end{list1}

Except work with industrial case studies, one extending work is to
eliminate the burden of the user by generating the
task-implication-graph automatically. A
\emph{task-implication-graph} specifies for each output port,
whether an incorrect function or an erroneous input will influence
the generated value. If there exists user-code (in general they
are task/control functions), then currently the translation into
languages acceptable by SMV is done by designers\footnote{In the
current version, the analyzer will approximate and imply that
every output of a given task is erroneous if there exists an
erroneous input, provided that no predefined
task-implication-graph is given}. For a task function with the set
of input ports $P_{in}$, the powerset of input ports and the set
inclusion relation $\sqsubseteq$ form a lattice
$(\mathbf{2}^{P_{in}},\sqsubseteq)$. Therefore, by using
techniques similar to monotone data-flow analysis
\cite{kam:1977:mdf}, a task-implication-graph can be constructed
automatically.

Finally, we will extend our modeling technique for the analysis
and verification of various models, e.g., SDF
\cite{girault:1999:hfs} (with hierarchical extensions) and FTDF
models with a predefined scheduling.


\appendix
\textbf{(Algorithmic Sketch of TestPortAbslute)} The mechanism of
\verb"TestPortAbsolute" consists of two rounds. Without loss of
generality, let $i$, $j$ and $k$ be the deployed sub-systems, and
let the algorithm describe behavior on machine $i$.
In the first round, it compares received values from other
machines with the local variable. For machine $i$ receiving from
sender $j$, if the value is different, then it evaluates $j$ be
faulty. The result of the evaluation is stored in an array
$(i_i,j_i,k_i)=\mathbf{2}^{3}$ and sent to other machines.

In the second round, as machine $i$ receives the array from
others, it performs a voting between elements in the set
$\{i_i,i_j,i_k\}$, $\{j_i,j_j,j_k\}$, and $\{k_i,k_j,k_k\}$. The
result of the voting finalizes the decision whether a machine is faulty or
not.

\bibliographystyle{abbrv}

\end{document}